\documentclass{llncs}
\usepackage[T1]{fontenc}
\usepackage{graphicx}
\usepackage{color}

\usepackage{cite}
\usepackage{amsmath,amssymb,amsfonts}
\usepackage[linesnumbered,ruled,vlined,noend]{algorithm2e}
\usepackage{booktabs}

\newcommand{\BO}[1]{O\left(#1\right)}

\newcommand{\dist}{\mbox{\rm dist}}
\newcommand{\divs}{\mbox{\rm div}}
\newcommand{\rank}{\mbox{\rm rank}}
\newcommand{\MST}{\mbox{\rm MST}}
\newcommand{\TSP}{\mbox{\rm TSP}}
\newcommand{\mindist}{d_{\rm min}}
\newcommand{\maxdist}{d_{\rm max}}
\newcommand{\minlv}{\ell_{\rm min}}
\newcommand{\maxlv}{\ell_{\rm max}}
\newcommand{\pts}{\mbox{\rm pts}}
\newcommand{\proc}[1]{\textsc{#1}}

\newtheorem{fact}{Fact}

\begin{document}
\title{Fully dynamic clustering and diversity maximization in doubling metrics}
%
%
\author{Paolo Pellizzoni\inst{1,2}\and
Andrea Pietracaprina\inst{2} \and
Geppino Pucci\inst{2}}
\authorrunning{F. Author et al.}
%
\institute{Max Planck Institute of Biochemistry\\
Martinsried, Germany \\
\email{pellizzoni@biochem.mpg.de} \and
Dept. of Information Engineering\\
University of Padova, Italy\\
\email{\{andrea.pietracaprina,geppino.pucci\}@unipd.it}}
\maketitle              
\begin{abstract}
We present approximation algorithms for some variants of center-based clustering and related problems in the fully dynamic setting, where the pointset evolves through an arbitrary sequence of insertions and deletions. Specifically, we target the following problems: $k$-center (with and without outliers), matroid-center, and diversity maximization. All algorithms employ a coreset-based strategy and rely on the use of the cover tree data structure, which we crucially augment to maintain, at any time, some additional information enabling the  efficient extraction of the solution for the specific problem. For all of the aforementioned problems our algorithms yield $(\alpha+\epsilon)$-approximations, where $\alpha$ is the best known approximation attainable in polynomial time in the standard off-line setting (except for $k$-center with $z$ outliers where $\alpha = 2$ but we get a $(3+\epsilon)$-approximation) and $\epsilon>0$ is a user-provided accuracy parameter. The analysis of the algorithms is performed in terms of the doubling dimension of the underlying metric. Remarkably, and unlike previous works, the data structure and the running times of the insertion and deletion procedures do not depend in any way on the accuracy parameter $\epsilon$ and, for the two $k$-center variants, on the parameter $k$. For spaces of bounded doubling dimension, the running times are dramatically smaller than those that would be required to compute solutions on the entire pointset from scratch. To the best of our knowledge, ours are the first solutions for the matroid-center and diversity maximization problems in the fully dynamic setting.

\end{abstract}

\section{Introduction}
Clustering, that is the task of partitioning a set of points according
to some similarity metric, is one of the fundamental primitives in
unsupervised learning and data mining, with applications in several
fields, such as bioinformatics, computer vision and recommender
systems \cite{Leskovec2020Mining-of-massi}.  A commonly used
formulation is the \emph{k-center} problem, where, given a set of
points from a metric space and a parameter $k$, one seeks to select
$k$ points as cluster centers so that the maximum distance from any
point to its closest center is minimized.  Since finding the optimal
solution is known to be NP-hard \cite{Gonzalez85}, in practice one has
to settle for approximate solutions.

In many practical applications, the set of points to be clustered is
not static, but evolves with time, and the clustering algorithms have
to cope with the insertion and deletion of arbitrary points
efficiently.  For example, in the context of social network analytics,
a large number of new profiles are created every second, therefore
adding new data to the pointset of interest, and the service provider
has to allow for the deletion of user data at any point in time, e.g.,
to comply with GDPR regulations. As another example, if one wishes to
track and cluster the currently watched contents from a video
streaming service, users starting and stopping to watch videos result
in a dynamically changing pointset, with thousands of updates
per second. Because of this ever-increasing need in efficient
algorithms that can cope with non-static data, in recent years there
has been a surge in  research efforts to develop
\textit{fully dynamic} clustering algorithms, whose main focus is on
the efficient handling
of arbitrary insertions and deletions
\cite{Chan2018Fully-dynamic-k, GoranciHLSS21,
  Bateni2023Optimal-Fully-D, Chan2023Fully-Dynamic-k}.
Several variants of the $k$-center problem have been intensely studied in recent
years, tailored at addressing specific limitation of the problem or at
adding additional constraints. Since $k$-center's objective function
involves a maximum, the optimal solution may be heavily affected by
points, dubbed \emph{outliers}, which are markedly distant from the
other ones.  Indeed, especially when data volumes are high, the
probability of observing outliers, such as from noisy or erroneous
measurements, is non negligible, and being able to cope with them is
of paramount importance.  A robust formulation of the problem,
referred to as \emph{$k$-center with $z$ outliers} has been introduced
in \cite{Charikar2001}, where the objective function is allowed to
disregard the $z$ most distant points from the selected centers.
Another studied variant of $k$-center is the \emph{matroid center}
problem \cite{ChenLLW16}, which, given a set of points and a matroid
defined on it, aims at finding a set of centers forming an independent
set of the matroid which minimizes the maximum distance of any point
from the closest center. This variant is extremely flexible and can be
employed to model a wide range of constraints (e.g., fairness) on the
solution \cite{ChiplunkarKR20}.  As for the standard $k$-center, the
matroid center problem admits a robust formulation with $z$ outliers
\cite{ChenLLW16}. Finally, a problem closely related to $k$-center is 
\emph{diversity maximization} \cite{AbbassiMT13}, whose goal is 
somehow dual with respect to the one of $k$-center, as $k$ points 
are to be selected  so that a suitable notion of diversity among them is maximized. 
This problem has a huge impact on information retrieval applications, 
where it is crucial that the information provided to the user be non-redundant. 

A widely-used notion of dimensionality for general metric spaces is the \emph{doubling dimension} 
\cite{Gupta2003Bounded-geometr, GottliebKK14}, which is defined as the minimum value
$D\geq 0$ such that any ball of any radius $r\geq 0$ can be covered by at most $2^D$ balls of radius $r/2$. 
In many real-world instances, the points of interest either belong to low-dimensional spaces
or lie on low-dimensional manifolds of the higher-dimensional metric space they belong to, 
and this property has been extensively exploited to obtain efficient clustering algorithms 
\cite{CeccarelloPPU17, PellizzoniPP20, GoranciHLSS21}. 
In this paper, we tackle the $k$-center problem, along with the
aforementioned variants to it, in the fully dynamic setting, for spaces
of bounded doubling dimension.

\subsection{Related work}
In the sequential setting, $k$-center admits 2-approximate algorithms,
such as Gonzalez's \cite{Gonzalez85}, but no
$(2-\epsilon)$-approximate ones unless P=NP. For the robust
formulation with $z$ outliers, there exists a simple, combinatorial  3-approximation
algorithm \cite{Charikar2001}, which has been extended to
weighted pointsets in  \cite{CeccarelloPP19}. More complex LP-based
2-approximate algorithms have been developed
\cite{Harris2019A-Lottery-Model, Chakrabarty2020The-Non-Uniform}, but
they are less amenable to practical implementations.  For what concerns the
matroid-center problem, \cite{ChenLLW16} provides a 3-approximate
algorithm for the standard formulation and a 7-approximate algorithm
for the formulation with outliers, which has been  improved to a
3-approximation in \cite{Harris2019A-Lottery-Model}.  The diversity
maximization problem admits several instantiations, depending on the
specific diversity function embodied in its objective, which are all
NP-hard but admit polynomial-time $\BO{1}$-approximation
algorithms. For an overview of such methods, we refer to
\cite{AbbassiMT13, CeccarelloPP20} and references therein.

In \cite{Chan2018Fully-dynamic-k}, the authors developed the first
fully dynamic $k$-center algorithm, which is able to return a
$(2+\epsilon)$-approximate solution under arbitrary insertions and
deletions of a non-adaptive adversary in general metrics. The
algorithm is randomized and has an (amortized) update time of $\BO{k^2
  \epsilon^{-1}\log\Delta}$, where $\Delta$ denotes the \emph{aspect
  ratio} of the pointset, namely the ratio between the largest and the
smallest distance between any two points.  This approach has been
recently improved in \cite{Bateni2023Optimal-Fully-D}, where an
algorithm with expected amortized $\BO{(k+\log n) \epsilon^{-1} \log
  \Delta \log n}$ update time is presented, where $n$ is the maximum
number of points at any time. It has to be remarked that both these
works make use of data structures for storing the dynamically changing
pointset, which are statically configured to deal with \emph{fixed}
values of $k$ and $\epsilon$. Answering queries for different
clustering granularities and/or accuracies would, in principle,
require building the data structure from scratch. Also both works use
data structures whose size, with respect to the number of points, is
superlinear by at least a factor $\BO{\log \Delta/\epsilon}$.
Recently, \cite{Chan2023Fully-Dynamic-k} presented a randomized
algorithm that returns a $14+\epsilon$ (bi-criteria) approximate
solution when discarding at most $(1+\lambda)z$ outliers. Again, their
data structure works for fixed $k$ and $\epsilon$, although it does
not depend on $\lambda$ and $z$, and requires superlinear size by a
factor $\BO{\log \Delta/\epsilon}$.

In \cite{GoranciHLSS21}, the authors propose a fully dynamic k-center
algorithm for points belonging to a low-dimensional space based on the
\emph{navigating net} data structure \cite{KrauthgamerL04}, which
affords insertions and deletions in $\BO{(1/\epsilon)^{\BO{D}}
  \log\Delta\log\log\Delta \log\epsilon^{-1}}$ time, where $D$ is the
doubling dimension of the metric space.  While their approach allows
clustering queries for arbitrary values of $k$, the data structure is
built for a specific value of the accuracy parameter $\epsilon$, and
must be rebuilt from scratch if a different accuracy $\epsilon'$ is
sought for.  Moreover, the data structure requires superlinear
space, by a factor
$\BO{(1/\epsilon)^{\BO{D}} \log\Delta\log\log\Delta \log\epsilon^{-1}}$.

Finally, $k$-center clustering has been tackled in other simpler
dynamic frameworks, such as in the insertion-only setting
\cite{Charikar1997Incremental-clu, CeccarelloPPU17,CeccarelloPP19,
  Kale19}, or in the sliding window setting
\cite{Cohen-Addad2016Diameter-and-k-, Borassi2019Better-sliding-,
  PellizzoniPP20, Pellizzoni2022k-Center-Cluste}. In these frameworks
though, the focus is usually on keeping the working memory sublinear
in the number of points rather than achieving fast update times.

\subsection{Our contributions}
In this paper, we present approximation algorithms for center-type
problems in the fully dynamic setting.  The core data structure at the
foundation of our approach is the \emph{cover tree}
\cite{BeygelzimerKL06}, which was originally designed for answering
nearest neighbor queries in low-dimensional metric spaces. We augment
such data structure to maintain the necessary information for solving
the problems at hand, under arbitrary insertion and deletions. The
augmented data structure is used to extract, at any point in time, a
\emph{coreset}, that is, a small subset of representative points that
can be used to obtain solutions of any desired accuracy on the whole
pointset for the problem at hand. Upon query requests, the solution is
returned by running a sequential algorithm on the extracted coreset,
which allows for fast execution times, independently on the current
number of points. Specifically, our contributions are the following.
\begin{enumerate}
\item 
We augment the cover tree data structure to maintain efficiently, in
each of its nodes, information about the subset of points stored in
its subtree, such as its cardinality or a maximal independent set
of the submatroid induced by
such subset (in case a matroid $M$ is defined over the whole pointset).
When no matroid information is maintained, the data structure takes
space linear in the number of points, and allows for insertion and
deletions in $\BO{2^{\BO{D}} \log\Delta}$ time, where $D$ is the
doubling dimension of the metric space, and $\Delta$ is the aspect
ratio of the current set of points. When maintaining matroid
information, the space requirements and the deletion times grow by a
factor at most $\rank(M)$.
\item 
We provide an iterative formulation of the algorithms used to maintain
the cover tree data structure, which, compared to the original
recursive formulation of \cite{BeygelzimerKL06}, affords simpler
correctness proofs and more efficient implementations. In fact, it has
been recently argued that the complexity analysis for the update
operations presented in the original paper has some flaws
\cite{Elkin2022Counterexamples}, which we fix by making it parametric
in $D$ and $\Delta$.
\item 
We devise a fully dynamic $(2+\epsilon)$-approximate algorithm for the
$k$-center problem. Unlike all aforementioned previous works, our data
structure allows to query for solutions for \emph{arbitrary values} of $k$ and
$\epsilon$. Also, the query time exhibits a linear dependency on $k$.
\item 
We devise a
$(3+\epsilon)$ approximation algorithm for the fully dynamic
$k$-center with $z$ outliers problem. Our method allows to choose
$k$, $\epsilon$, and $z$ at query time. We remark that the only
previously available algorithm can only return a $(14+\epsilon)$
bi-criteria solution (i.e. with an additional slackness on the number
of outliers), and requires to fix $k$ and $\epsilon$ beforehand.
\item 
We present the first
fully dynamic algorithm for the matroid center problem. Our algorithm
return a $(3+\epsilon)$ approximate solution.
\item 
We present the first
fully dynamic algorithms for diversity maximization. Our algorithms
return $(\alpha_{\rm div}+\epsilon)$ approximate solutions, where
$\alpha_{\rm div}$ is the best approximation factor achievable by a
sequential algorithm on the variant of the problem at hand.  
\end{enumerate}

An important feature of our algorithms is that they are fully
oblivious to $D$, in the sense that the actual value of this parameter
only influences the analysis but are not needed for the algorithms to
run. This is a very desirable feature, since, in practice, this value
is difficult to estimate.

The rest of the paper is organized as follows.
Section~\ref{sec:preliminaries} provides the formal definition of the
problems and of the theoretical tools we use in the analysis.
Section~\ref{sec:augcovertrees} is dedicated to describing the
structure of the novel augmented cover tree data structure, while
Section~\ref{sec:maintain} details how to maintain such data structure
efficiently.  In Section~\ref{sec:approx}, we present the
approximation algorithms, which rely on the cover tree data
structure. Section~\ref{sec:conclusion} concludes the paper with some
final remarks.

\section{Preliminaries} \label{sec:preliminaries}
This section formally defines the problems studied in this paper, and states some important technical facts.

\subsection{(Robust) $k$-center problem}
Consider a metric space $(U,\dist)$ and a set $S \subseteq U$ of $n$ points. 
For any $p \in U$ and any subset  $C \subseteq S$, 
we use the notation $\dist(p,C) = \min_{q \in C} \dist(p,q)$,
and define the \emph{radius of $C$ with respect to $S$} as
\[
r_C(S) = \max_{p \in S} \dist(p,C).
\] 
For a positive integer $k < n$, the \emph{$k$-center} problem requires
to find a subset $C \subseteq S$ of size at most $k$ which minimizes
$r_C(S)$. The points of the solution $C$ are referred to as
\emph{centers}.  Note that $C$ induces a partition of $S$ into $|C|$
clusters, by assigning each point to its closest center (with ties
broken arbitrarily).  We denote the radius of the optimal solution by
$r_{k}^*(S)$. The popular seminal work by Gonzalez \cite{Gonzalez85}
presents a 2-approximation sequential algorithm for the $k$-center
program, based on the simple, $\BO{nk}$-time greedy strategy that
selects the first center arbitrarily and each subsequent center as the
point with maximum distance from the set of previously selected
ones. The author also shows that, in general metric spaces, it is
impossible to achieve an approximation factor $2-\epsilon$, for any
fixed $\epsilon >0$, unless P = NP.

The algorithms presented in this paper crucially rely on confining the
computation of the solution on a succinct \emph{coreset} $T$
efficiently extracted from the (possibly large) input $S$, which
contains a close enough ``representative'' for each point in $S$. The
quality of a coreset $T$ is captured by the following definition.

\begin{definition} \label{def:coreset}
Given a pointset $S$ and a value $\epsilon >0$, a subset $T \subseteq
S$ is an \emph{$(\epsilon,k)$-coreset for $S$} $($w.r.t. the $k$-center
problem$)$ if $r_T(S) \leq \epsilon r^*_k(S)$.
\end{definition}

In real world applications, large datasets often include noisy points
which, if very distant from all other points, may severely distort the
optimal center selection. To handle these scenarios, the following
important generalization of the $k$-center problem has been
considered. For positive $k,z < n$, the \emph{$k$-center problem with
  $z$ outliers} (also referred to as \emph{robust $(k,z)$-center})
requires to find a subset $C \subseteq S$ of size $k$ minimizing
$r_C(S-Z_C)$, where $Z_C$ is the set of $z$ points in $S$ with the
largest distances from $C$, which are regarded as outliers to be
discarded from the clustering. We denote the radius of the optimal
solution of this problem by $r_{k, z}^*(S)$.  Observe that the
$k$-center problem with $z$ outliers reduces to the $k$-center problem
for $z=0$. Also, it is straightforward to argue that
\begin{equation} 
r_{k+z}^*(S) \le r_{k,z}^*(S).
\label{eq:radius-relation}
\end{equation}

A well known 3-approximation sequential algorithm for the $k$-center
problem with $z$ outliers, which runs in $\BO{kn^2 \log n}$ time was
devised in \cite{Charikar2001}. In this work, we will make use of a
more general formulation of the problem, referred to as \emph{weighted
  $k$-center with $z$ outliers}, where each point $p \in S$ carries a
positive integer weight $w(p)$, and the desired set $C$ of $k$ centers
must minimize $r_C(S-Z_C)$, where $Z_C$ is the set of points with the
largest distances from $C$, of maximum cardinality and aggregate
weight at most $z$. 

\subsection{Matroid center problem}
Another variant of the $k$-center problem 
requires the solution $C$ to satisfy an additional constraint,
specified through a matroid. A
\emph{matroid}~\cite{Oxley06} on a pointset $S$ is a pair $M =
(S,I)$, where $I$ is a family of subsets of $S$, called
\emph{independent sets}, satisfying the following properties: (i) the
empty set is independent; (ii) every subset of an independent set is
independent (\emph{hereditary property}); and (iii) if $A,B \in I$ and
$|A| > |B|$, then there exists $x \in A\setminus B$ such that $B \cup
\{x\} \in I$ (\emph{augmentation property}). An independent set is
\emph{maximal} if it is not properly contained in another independent
set. Given a matroid $M=(S,I)$, the \emph{matroid center problem} on
$M$ requires to determine an independent set $C \in I$ minimizing the
radius $r_C(S)$. We let $r^*(M) = \min_{C \in I} r_C(S)$ to denote the
radius of the optimal solution. The augmentation property ensures
that all maximal independent sets of $M$ have the same size, which is
referred to as \emph{rank} of $M$, denoted by $\rank(M)$.
It is easy to argue that $r^*_{{\rm rank}(M)}(S) \leq r^*(M)$, 
since each solution to the matroid center problem on is also a solution to
the $k$-center problem with $k = \rank(M)$.

As customary in previous works
\cite{AbbassiMT13,Kale19,CeccarelloPP20}, throughout the paper we
assume that a constant-time oracle is available to check the
independence of any subset of $S$.  A combinatorial 3-approximation
algorithm for the matroid center problem in general metrics is
presented in \cite{ChenLLW16}, which requires time
polynomial in $|S|$ and $\rank(M)$.

Some important structural properties of matroids will be used in the
derivations of our results. It is easily seen that for any subset $S'
\subseteq S$, $M'=(S',I')$, where $I' \subseteq I$ is the restriction
of $I$ to subsets of $S'$ is also a matroid.
An \emph{extended augmentation property}, stated in the following fact, 
was proved in \cite[Lemma 2.1]{PietracaprinaPS20}
(see also \cite{Kale19}).

\begin{fact}\label{fact:extAug}
Let $M=(S,I)$ be a matroid. Consider an independent set $A \in I$, a
subset $S' \subseteq S$, and an independent set $B \subseteq S'$ which
is maximal within the submatroid $M'=(S',I')$.  If there exists $y \in
S'\setminus A$ such that $A \cup \{y\} \in I$, then there exists $x
\in B \setminus A$ such that $A \cup \{x\} \in I$.
\end{fact}

The next fact can be easily proved as a consequence of the 
extended augmentation property.
\begin{fact}\label{fact:combination}
Let $M=(S,I)$ be a matroid, and let $S_1, \ldots, S_h$ be a partition of $S$ into $h$ disjoint subsets. If $A_1\subseteq S_1, \ldots, A_h \subseteq S_h$ are maximal independent sets of the submatroids $M_1= (S_1,I_1), \ldots, M_h=(S_h, I_h)$, then $\cup_{I=1}^h A_i$ contains a maximal independent set of $M$. 
\end{fact}
 
Different matroids are associated to different semantics of the
constraint on the centers. An important instantiation is the case of
the \emph{partition matroid} $M_P=(S,I_P)$, where each point in $S$ is
associated to one of $m \leq k$ of categories, and $I_P$
consists of all subsets with at most $k_i$ points of the $i$-th
category, with $\sum_{i=0}^m k_i = k$. For instance, 
this matroid can be employed to model fairness constraints
\cite{KleindessnerAM19,ChiplunkarKR20}.

\subsection{Diversity maximization}
Let $\divs: 2^S \rightarrow \mathbb{R}$ be a \emph{diversity function}
that maps any subset $X \subset S$ to some non-negative real number.
For a specific diversity function $\divs$ and a positive integer $k
\leq n$, the goal of the \emph{diversity maximization problem} is to
find a set $C \subseteq S$ of size $k$ that maximizes $\divs(C)$.  We
denote the optimal value of the objective function as $\divs^*_{k}(S)
= \max_{C \subseteq S, |C| = k} \divs(C)$. In this paper, we will
focus on several variants of the problem based on different diversity
functions amply studied in the previous literature, which are
reported in Table~\ref{tab:diversity-notions}.
All of these variants are known to be NP-hard, and
Table~\ref{tab:diversity-notions} also lists the best known 
approximation ratios attainable in polynomial time (see
\cite{AbbassiMT13,CeccarelloPPU17} and references therein).

\newcommand{\definitionTableVerticalSpacing}{\rule{0pt}{10pt}}
\begin{table}[t]
  \centering
  \begin{tabular}{l@{\hskip 5pt} l@{\hskip 1pt} c}
    \toprule
    \shortstack{Problem\\~}
     & \shortstack{Diversity\\measure $\divs(X)$}
     & \shortstack{Sequential\\approx. $\alpha_{\rm div}$}
    \\
    \midrule
    remote-edge
      & $\min_{p, q\in X} d(p, q)$
      & 2 
    \\
    \definitionTableVerticalSpacing%
    remote-clique
      & $\sum_{p, q\in X} d(p, q)$ 
      & 2 
    \\
    \definitionTableVerticalSpacing%
    remote-star 
      & $\min_{c\in X}\sum_{q\in X\setminus\{c\}} d(c, q)$ 
      & 2
    \\
    \definitionTableVerticalSpacing%
    remote-bipartition
     & 
$\displaystyle\min_{Q\subset X, |Q|=\lfloor|X|/2\rfloor}{\textstyle\sum}_{q\in Q, z\in X \setminus Q} d(q, z)$ 
     & 3
    \\
    \definitionTableVerticalSpacing%
    remote-tree 
      & $w(\MST(X))$ 
      & 4
    \\
    \definitionTableVerticalSpacing%
    remote-cycle 
      & $w(\TSP(X))$
      & 3
    \\
    \bottomrule
\hspace*{0.2cm}
  \end{tabular}
  \caption{%
Diversity functions considered in this paper.  $w(\MST(X))$ (resp.,
$w(\TSP(X))$) denotes the minimum weight of a spanning tree (resp.,
Hamiltonian cycle) of the complete graph whose nodes are the points of
$X$ and whose edge weights are the pairwise distances among the
points. The last column lists the best known approximation factor.
} \label{tab:diversity-notions}
\end{table}

\subsection{Doubling dimension}
We will relate the performance of our algorithms to the dimensionality
of the data which, for a general metric space $(U, \dist)$, 
can be captured by the
notion of doubling dimension, reviewed below. For any $p \in U$ 
and $r >0$, the \emph{ball of radius $r$ centered at $p$}, denoted as
$B(p,r)$, is the subset of all points of $U$ at distance at most $r$
from $p$. The \emph{doubling dimension} of $U$ is the minimum value
$D$ such that, for all $p \in U$, any ball $B(p,r)$ is contained in
the union of at most $2^D$ balls of radius $r/2$ centered at points of
$U$. In recent years, the notion of doubling dimension has been used extensively for a variety of applications 
(e.g., see \cite{GottliebKK14,CeccarelloPPU16,CeccarelloPPU17,PellizzoniPP20} 
and references therein).
The following fact is a simple consequence of the definition. 
\begin{fact} \label{fact:doubling}
Let $X$ be a set of points from a metric space of doubling dimension $D$, 
and let $Y \subseteq X$
be such that any two distinct points $a,b \in Y$ have pairwise
distance $\dist(a,b) > r$.  Then 
for every $R \geq  r$ and any point $p \in X$, we have $|B(p,R) \cap Y|
\leq 2^{\lceil \log_2(2R/r) \rceil \cdot D} \leq (4R/r)^D$. If $R/r = 2^i$, then
 the bound can be lowered to $|B(p,R) \cap Y| \leq (2R/r)^D$.
\end{fact}

\section{Augmented Cover Trees} \label{sec:augcovertrees}
Let $S$ be a set of $n$ points from a metric space $(U,\dist)$ of
doubling dimension $D$.  Our algorithms employ an augmented version of
the cover tree data structure, defined in \cite{BeygelzimerKL06},
which, for completeness, we review below.  Conceptually, a \emph{cover
  tree} $T$ for $S$ is an infinite tree, where each node corresponds
to a single point of $S$, while each point of $S$ is associated to one
or more nodes.  The levels of the tree are indexed by integers,
decreasing from the root towards the leaves.  For $\ell \in
(-\infty,+\infty)$, we let $T_{\ell}$ to be the set of nodes of level
$\ell$, and let $\pts(T_{\ell})$ be the points associated with the
nodes of $T_{\ell}$, which are required to be distinct.  For each node
$u \in T$ we maintain: its associated point ($u.\mbox{point}$); a
pointer to its parent ($u.\mbox{parent}$); the list of pointers to its
children ($u.\mbox{children}$); and the level of $u$ in $T$
($u.\mbox{level}$).  For brevity, in what follows, for any two nodes
$u,v \in T$ (resp., any point $p\in U$ and node $u\in T$), we will use
$\dist(u,v)$ (resp., $\dist(p, u))$ to denote
$\dist(u.\mbox{point},v.\mbox{point})$ (resp.,
$\dist(p,u.\mbox{point})$).
The tree must satisfy the following properties. For every level $\ell$:
\begin{enumerate}
\item 
\label{CTp1}
$\pts(T_{\ell}) \subseteq \pts(T_{\ell-1})$;
\item 
\label{CTp2}
for each $u \in T_{\ell}$, $\dist(u, u.{\rm parent}) \leq 2^{\ell+1}$;
\item 
\label{CTp3}
for all $u, v \in T_{\ell}$, $\dist(u, v) > 2^{\ell}$.
\end{enumerate}
For every $p \in S$, let $\ell(p)$ denote the largest index such that
$p \in \pts(T_{\ell(p)})$.  The definition implies that for every
$\ell \leq \ell(p)$, $p \in \pts(T_{\ell})$, and the node $u$ in
$T_{\ell}$ with $u.\mbox{point}=p$ is a \emph{self-child} of itself,
in the sense that $u.\mbox{parent}.\mbox{point} = p$, since for every
other $v \in T_{\ell+1}$, with $v.\mbox{point} \neq p$, $\dist(u,v) >
2^{\ell+1}$.  Let $\mindist$ and $\maxdist$ denote, respectively, the
minimum and maximum distances between two points of $S$, and define
$\Delta = \maxdist/\mindist$ as the \emph{aspect ratio} of $S$. It can
be easily seen that for every $\ell < \log_2 \mindist$,
$\pts(T_{\ell})=S$, and for every $\ell \geq \log_2 \maxdist$,
$|\pts(T_{\ell})|=1$. We define $\minlv$ (resp., $\maxlv$) as the
largest (resp., smallest) index such that $\pts(T_{\ell})=S$ (resp.,
$|\pts(T_{\ell})|=1$), and note that every node $u$ in a level $T_j$
with $j > \maxlv$ or $j \leq \minlv$ has only the self-child in
$u.\mbox{children}$. Therefore, we will consider only the portion of
the tree constituted by the levels $T_{\ell}$ with $\ell \in
[\minlv,\maxlv]$ and will regard the unique node $r\in T_{\maxlv}$ as
the root of the tree.  The above observations imply that the number of
levels in this portion of the tree is $\BO{\log \Delta}$.

Cover trees have been initially proposed as data structures for
efficiently retrieving nearest neighbors. This feature, which will also be 
exploited in our context, crucially relies on the notion of \emph{cover set}. 
For any point $p$ from the metric space $U$, and for every
index  $\ell \leq \maxlv$  the \emph{cover set} 
(\emph{for $p$ at level $\ell$})
$Q^p_{\ell} \subseteq T_{\ell}$ is defined in an inductive way as follows:
\begin{equation} \label{def-coverset}
\begin{array}{lll}
Q^p_{\maxlv} & = & \{ r\} \\[0.2cm]
Q^p_{\ell} & = & \{u \in T_{\ell} \; : \; 
u.\mbox{parent} \in Q^p_{\ell+1} \wedge \dist(u,p) \leq 2^{\ell+1}\}
\;\;\; \mbox{ for } \ell < \maxlv
\end{array}
\end{equation}

The next lemma shows that cover sets for a point $p$ include, at
each level of the tree, all points somewhat ``close'' to $p$, at a scale prescribed by the level.
\begin{lemma} \label{lem:coverdist}
For every point  $p$ from $(U,\dist)$, $\ell \leq \maxlv$,
and $u \in T_{\ell}-Q^p_{\ell}$, 
we have $\dist(p,u) > 2^{\ell+1}$. 
\end{lemma}
\begin{proof}
The proof proceeds by (backward) induction on $\ell$. For the  base case $\ell =\maxlv$,  the statement holds vacuously.
Assume now that the statement holds at level $\ell +1$, and let $u \in T_{\ell}-Q^p_{\ell}$. 
The statement immediately follows when  
$u.\mbox{parent} \in Q^p_{\ell+1}$. Otherwise, since  $u.\mbox{parent} \in Q^p_{\ell+1} -T_{\ell+1}$,  by the inductive hypothesis it must be $\dist(p, u.\mbox{parent}) > 2^{\ell+2}$, whence $\dist(p,u) \geq \dist(p, u.\mbox{parent}) - \dist(u, u.\mbox{parent})  > 2^{\ell+1}$.
\end{proof}
We now relate the size of the cover sets to the doubling dimension $D$ of $(U,\dist)$.
\begin{lemma} \label{lem:coversize}
For every point  $p\in U$ and $\ell \leq \maxlv$,
we have that 
\begin{enumerate}
\item
$|Q^p_{\ell}| \leq 4^D$.
\item 
$\sum_{u\in Q^p_{\ell}} |u.\mbox{\rm children}| \leq 12^D$.
\end{enumerate}
\end{lemma}
\begin{proof}
For each $u_1, u_2 \in T_\ell$, we have that $\dist(u_1, u_2) >
2^\ell$. Also, $\pts(Q^p_{\ell})\subseteq B(p, 2^{\ell+1})\cap
T_\ell$. Then, by Fact~\ref{fact:doubling}, it follows that
$|Q^p_{\ell}| \leq (2 \cdot 2^{\ell+1} / 2^\ell)^D = 4^D$. Moreover,
for each $u \in Q^p_\ell$ and each $u' \in u.{\rm children}$, we have
that $\dist(u', p) \leq \dist(u', u) + \dist(u, p) \leq 2^{\ell} + 2^{\ell+1}
= 3\cdot 2^{\ell}$, whence $u'.\mbox{point} \in B(p, 3\cdot 2^\ell) \cap \pts(T_{\ell -1})$.
Then, again by Fact~\ref{fact:doubling}, it follows that $|\{ u' \in
u.{\rm children } $ s.t. $ u \in Q^p_{\ell} \}| \leq (2\cdot 3 \cdot
2^{\ell} / 2^{\ell-1})^D = 12^D$.
\end{proof}

We represent the entire tree  $T$  by  a pointer to the root, and recall that $r.\mbox{level} = \maxlv$. Note that this naive representation of $T$,
referred to as \emph{implicit representation} in
\cite{BeygelzimerKL06}, requires $\BO{n \log\Delta}$ space.
In order to save space, a more compact representation $T$, referred to as
\emph{explicit representation} in \cite{BeygelzimerKL06}, is used,
where chains of 1-child nodes, which correspond to instances of
the same point, are coalesced. More precisely, this representation
only maintains explicitly nodes that have children other than the
self-child.
Therefore, any maximal chain of nodes in $T_i, T_{i-1}, \ldots,
T_{i-j}$ all associated to the same point $p$ and such that each of
them (up to Level $T_{i-j+1}$) is only parent to the self-child, is represented
through the explicit node $u\in T_i$  with
$u.\mbox{children}$ set to the list of children of the
node of the chain in $T_{i-j}$. It is easy to argue that this compact representation
takes only $\BO{n}$ space.  
It is important to observe that, given a point $p\in U$, the cover
sets $Q^p_{\ell}$ for $\ell \leq \maxlv$ can be constructed in a
top-down fashion from the explicit representation of $T$ by simply
recreating the contracted chains of implicit nodes. In our algorithms,
in each cover set $Q^p_\ell$, an implicit node $v$ will be represented
by the explicit ancestor $u$ associated to the head of the chain
containing $v$. By Lemma~\ref{lem:coversize}, $Q^p_{\ell}$ can be constructed
from $Q^p_{\ell+1}$ in $\BO{12^D}$ time.

\subsection{Augmenting the basic structure}
In this subsection, we show how to augment the cover tree data
structure so to maintain at its nodes two additional data fields,
which will be exploited in the target applications presented
later. Suppose that a matroid $M=(U,I)$ is defined over the universe
$U$.  An \emph{augmented cover tree} $T$ for $S$ (with respect to $M$)
is a cover tree such that each node $U$ stores the following two
additional fields: a positive weight $u.\mbox{\em weight}=|S_u|$, where
$S_u$ is the subset of points of $S$ associated with nodes in the subtree rooted at node
$u$, and a set of points $u.\mbox{\em mis}$, which is a maximal
independent set of the submatroid $M_u =(S_u,I_u)$ 
where  $I_u$ is the restriction of $I$ to the subsets of $S_u$.  The
size of $T$ becomes $\BO{n\cdot m}$, where $m$ denotes the size
of a maximum independent set of $S$.  For the applications where the
matroid information is not needed, the fields $u.\mbox{mis}$ will
be always set to null (as if $M=(U,\emptyset)$), and, in this case,
the size of $T$ will be $\BO{n}$.

\section{Dynamic maintenance of augmented cover trees} \label{sec:maintain}
Let $T$ be an augmented cover tree for the set $S$ of $n$ points.  In
this section, we show how to update $T$ efficiently when a point
$p$ is added to or deleted from $S$.

\subsection{Insertion} \label{sec:insertion}
Let $p$ be a new point to be inserted in $T$.  The insertion of $p$ is
accomplished as follows.  First, if $p$ is very far from the root $r$,
namely $\dist(p,r) > 2^{\maxlv}$, then both $\maxlv$ and
$r.\mbox{level}$ are increased to $\lfloor \log_2 \dist(p,r)
\rfloor$. Then, an explicit node $u$ is created with
$u.\mbox{point}=p$, $u.\mbox{weight}=1$ and $u.\mbox{mis}=\{p\}$.  In
order to determine the level $\ell(p)$ where $u$ must be placed, all
cover sets $Q^p_{\ell}$ are computed, as described above, for every
$\ell \in [\bar{\ell},\maxlv]$, where $\bar{\ell}$ is the largest
index in $(-\infty,\maxlv]$ such that $Q^p_{\bar{\ell}} =
  \emptyset$. Note that such empty cover set must exist and it is easy
  to see that $\bar{\ell} \geq \lceil \log_2 \dist(p,S) \rceil-2$.
  Then, $u.\mbox{level}$ is set to the smallest index $\ell(p) \geq
  \bar{\ell}$ such that $\dist(p,\pts(Q^p_{\ell(p)})) > 2^{\ell(p)}$
  and $\dist(p,\pts(Q^p_{\ell(p)+1})) \leq 2 ^{\ell(p)+1}$. At this
  point, an arbitrary node $v \in Q^p_{\ell(p)+1}$ such that
  $\dist(p,v) \leq 2^{\ell(p)+1}$ (which must exists for sure) is
  determined.  Let $q = v.\mbox{point}$.  If $v$ has no explicit
  self-child at level $\ell(p)$, a new node $w$ with
  $w.\mbox{point}=q$, $w.\mbox{level}=\ell(p)$, and
  $w.\mbox{children}=v.\mbox{children}$, is created, and
  $v.\mbox{children}$ is set to $\{u,w\}$. If instead such an explicit
  self-child $w$ of $v$ exists at level $\ell(p)$, then $u$ is simply
  added as a further child of $v$. Finally, the path from the newly
  added node $u$ to $r$ is traversed, and for every ancestor $v$ of
  $u$, $v.\mbox{weight}$ is increased by 1 and $p$ is added to the
  independent set $v.\mbox{mis}$, if $v.\mbox{mis} \cup \{p\}$ is
  still an independent set.
Algorithm $\proc{Insert}(p,T)$ in the appendix, provides the pseudocode for the above 
procedure.

\begin{theorem}\label {thm:insert}
Let $T$ be an augmented cover tree for a set $S$ of
$n$ points, with respect to a matroid $M=(U,I)$. 
The insertion algorithm described above yields an augmented cover tree
for $S \cup \{p\}$ in time $\BO{12^D \log \Delta}$
where $D$ is the doubling dimension of the metric space 
and $\Delta$ is the
aspect ratio of $S$.
\end{theorem}
\begin{proof}
It is easy to see that the insertion algorithm
enforces, for every level $\ell$,
Properties \ref{CTp1}, \ref{CTp2}, \ref{CTp3} of the definition of
cover tree, restricted to the nodes in $Q^p_{\ell}$ plus the new node
created for $p$ (for $\ell = \ell(p)$).
These properties immediately extend to
the entire level $\ell$ by virtue of Lemma~\ref{lem:coverdist}.
For what concerns the update of the $.\mbox{weight}$ and
$.\mbox{mis}$ fields, correctness is trivially argued for the $.\mbox{weight}$
fields,  while Fact~\ref{fact:combination} ensures
correctness of the updates of the $.\mbox{mis}$ fields. 
 can be argued as for the insertion algorithm. 
The complexity
bounds follow by observing that there are $\BO{\log \Delta}$ levels in
the explicit representation of $T$ and that, at each such level
$\ell$, the algorithm performs work linear in the number of children
of $Q^p_{\ell}$, which are at most $12^D$, by
Lemma~\ref{lem:coversize}.
\end{proof}

\subsection{Deletion}
Let $p\in S$ be the point to be removed. We assume that $p$ is not the
only point in $S$, otherwise the removal is trivial.  The deletion of
$p$ is accomplished as follows.  In the first, top-down phase, all
cover sets $Q^p_{\ell}$ are computed, for every $\ell \in
[\bar{\ell},\maxlv]$, where $\bar{\ell} \leq \maxlv$ is the 
level of the  leaf node  corresponding to $p$ in the explicit tree. 
Also, a list $R_{\bar{\ell}}$ of explicit
  nodes at level $\bar{\ell}$ to be relocated is created and
  initialized to the empty list.  In the second, bottom-up phase, the
  following operations are performed iteratively, for every $\ell =
  \bar{\ell}, \bar{\ell}+1, \ldots, \maxlv-1$ (the case $\ell=\maxlv$
  will be treated separately).
\begin{itemize}
\item
If $Q^p_{\ell}$ contains a node $u$ with $u.\mbox{point}=p$ and
$u.\mbox{level}=\ell$, the following additional
operations are performed.  Let $v$ be the parent of $u$, and observe
that all children of $v$ must also be explicit nodes at level $\ell$.
If $u$ is the self child of $v$ (i.e., $v.\mbox{point}=u.\mbox{point}=p$)
$u$ is removed from $T$, 
$u$'s siblings are detached from $v$ 
and added to $R_{\ell}$ ($v$ will be later removed
at iteration $v.\mbox{level}$). If instead $u$ is not the self
child of $v$, but it is the only child of $v$ besides the self-child,
$u$ is removed from $T$ and 
$v$ and its self-child are merged together in the explicit tree.
\item
An empty list $R_{\ell+1}$ is created. Then, 
$R_{\ell}$ is scanned sequentially, and, for every $w \in R_{\ell}$,
a node $w' \in Q^p_{\ell+1} \cup R_{\ell+1}$  is searched for such that
$d(w,w') \leq 2^{\ell+1}$. If no such node exists,
then $w$ is added to $R_{\ell+1}$, raising $w.\mbox{level}$ to
$\ell+1$. Otherwise, if $w'$ is found, it becomes parent of $w$ as follows. 
If $w'$ is internal and its children  
are at level $\ell$, $w$ is simply added as a further child. 
Otherwise, 
a new explicit node $z$, is created with
$z.\mbox{point}= w'.\mbox{point}$, $z.\mbox{level}=\ell$,
and  $z.\mbox{children}=w'.\mbox{children}$, and
$w'.\mbox{children}$ is set to $\{z,w\}$. 
\item
For all nodes in $w \in Q^p_{\ell+1} \cup R_{\ell+1}$ with $w.\mbox{level} =\ell+1$, their $.\mbox{weight}$  and $.\mbox{mis}$  fields are updated based on the values of the
corresponding fields of their children. The update of the $.\mbox{weight}$ fields is 
straightforward, while, based on Fact~\ref{fact:combination}, the update of the $.\mbox{mis}$  field of one such node $w$ can be accomplished by computing a maximal independent set 
in the union of the elements of the $.\mbox{mis}$ fields of $w$'s children. 
\end{itemize}
Once the above operations are performed up to $\ell = \maxlv-1$, 
the following cases arise for level $\maxlv$. Consider first the case
$p \neq r.\mbox{point}$. If $R_{\maxlv} = \emptyset$, we simply update $\maxlv$ and $r.\mbox{level}$
to 1 plus the level of the children of $r$; otherwise 
($R_{\maxlv} \neq \emptyset$), a new root $r_{\rm new}$ is created with
$r_{\rm new}.\mbox{point}=r.\mbox{point}$, 
$r_{\rm new}.\mbox{level}=\maxlv+1$,
and $r_{\rm new}.\mbox{children} = \{r\} \cup R_{\maxlv}$. Also,
$\maxlv$ is incremented by 1. Instead, in case $p = r.\mbox{point}$,
it is easy to see that $R_{\maxlv}$ cannot be empty, since it must contain
for sure the children of $r$. If $R_{\maxlv}$ contains
only one node, say $v$, then $v$ becomes the new root, and we update, 
$\maxlv$ and $v.\mbox{level}$
to 1 plus the level of the children of $v$; otherwise ($|R_{\maxlv}|>1$),
we select an arbitrary $v \in R_{\maxlv}$, 
a new root $r_{\rm new}$ is created with
$r_{\rm new}.\mbox{point}=v.\mbox{point}$, 
$r_{\rm new}.\mbox{level}=\maxlv+1$,
and $r_{\rm new}.\mbox{children} = R_{\maxlv}$.
Also,
$\maxlv$ is incremented by 1.
Finally, whenever level $\maxlv$ is incremented by 1 and, consequently,
a new root node is created, for this node the 
$.\mbox{weight}$ and $.\mbox{mis}$ fields are updated based on the values of the
corresponding fields of their children, as described above. 
Algorithm $\proc{Delete}(p,T)$ in the appendix, provides the pseudocode for the
above procedure.

\begin{theorem}
Let $T$ be an augmented cover tree for a set $S$ of $n$ points, with
respect to a matroid $M=(U,I)$. The deletion algorithm described above
yields an augmented cover tree for $S-\{p\}$ in time $\BO{(16^D+12^D\rank(M)) \log \Delta}$ where $D$ is the doubling dimension of the metric space and
$\Delta$ is the aspect ratio of $S$.
\end{theorem}
\begin{proof}
It is easy to see that the bottom-up phase of the deletion algorithm
enforces, for every level $\ell$,
Properties \ref{CTp1}, \ref{CTp2}, \ref{CTp3} of the definition of
cover tree, restricted to the nodes in  $Q^p_{\ell} \cup R_{\ell}$.
These properties immediately extend to
the entire level $\ell$ by virtue of Lemma~\ref{lem:coverdist}.
Finally, correctness of the update of the $.\mbox{weight}$ and
$.\mbox{mis}$ fields can be argued as for the insertion algorithm. 
For what concerns the complexity bound, first observe that for every level $\ell$,
the nodes in $Q^p_{\ell} \cup R_{\ell}$ represent  the new coversets $\hat{Q}^p_{\ell}$ associated to 
$p$ after its deletion from $T$, thus Lemma~\ref{lem:coversize} holds. As a consequence, 
the work needed to process the nodes in $R_\ell$ is  $\Theta(|R_\ell|(|Q^p_{\ell+1} \cup R_{\ell+1}|) = O(16^D)$, 
while, by Fact~\ref{fact:combination}, recreating the $.\mbox{weight}$ and $.\mbox{mis}$ fields for all nodes in $Q^p_{\ell+1} \cup R_{\ell+1}$ is
upper bounded by the number of their children multiplied by $(1+\mbox{rank}(M))$. The final bound follows 
 by observing that there are $\BO{\log \Delta}$ levels in
the explicit representation of $T$.
\end{proof}

\section{Extracting solutions from the augmented cover tree} \label{sec:approx}
We show how to employ the augmented cover tree
presented in the previous section to extract accurate solutions to
the various $k$-center related and diversity maximization problems introduced
in Section~\ref{sec:preliminaries}. For all these problems, we rely on the
extraction from the cover tree 
of a small $(\epsilon,k)$-coreset (see Definition~\ref{def:coreset}), 
for suitable values of $\epsilon$ and $k$.

Let $T$ be an augmented cover tree for a set $S$ of $n$ points
from a metric space of doubling dimension $D$. Given $\epsilon$ and $k$, 
an  $(\epsilon,k)$-coreset for $S$ can be constructed as follows. 
Let $T_{\ell(k)}$ be the level of largest index 
(in the implicit representation of $T$)
such that $|T_{\ell(k)}| \leq k$ and $|T_{\ell(k)-1}| > k$. Then, define
\begin{equation} \label{eq:ellstar}
\ell^*(\epsilon,k) =  \max\{\minlv,\ell(k)-\lceil \log_2 (8/\epsilon) \rceil\}.
\end{equation}
(For ease of notation, in what follows we shorthand
$\ell^*(\epsilon,k)$ with $\ell^*$ whenever the parameters are clear from the
context.)  We have:
\begin{lemma} \label{lem:coreset}
The set of points $\pts(T_{\ell^*})$ is an $(\epsilon,k)$-coreset for
$S$ of size at most $k(64/\epsilon)^D$ and can be constructed in time
$\BO{k((64/\epsilon)^D+\log \Delta)}$.
\end{lemma}
\begin{proof}
We first show that $\pts(T_{\ell^*})$ is an $(\epsilon,k)$-coreset for
$S$.  If $\ell^*=\minlv$, we have $\pts(T_{\ell^*})=S$, so the
property is trivially true. Suppose that $\ell^* = \ell(k)-\lceil
\log_2 (8/\epsilon)\rceil > \minlv$ and consider an arbitrary point $p \in
S$. There must exist some node $u \in T_{\minlv}$ such that $p =
u.\mbox{point}$. Let $v$ be the ancestor of $u$ in $T_{\ell^*}$.  By
the properties of the cover tree we know that $\dist(v.\mbox{point},p)
\leq 2^{\ell^*+1} \leq (\epsilon/4) 2^{\ell(k)}$.  Also, all pairwise
distances among points of $\pts(T_{\ell(k)-1})$ are greater than
$2^{\ell(k)-1}$. However, since $|T_{\ell(k)-1}| \geq k+1$, there must
be two points $q,q' \in \pts(T_{\ell(k)-1})$ which belong to the same
cluster in the optimal $k$-center clustering of $S$. Therefore,
$2^{\ell(k)-1} < \dist(q,q') \leq 2r^*_k(S)$, which implies that
$2^{\ell(k)} \leq 4r^*_k(S)$. Putting it all together, we get that for
any $p \in S$, $\dist(p,\pts(T_{\ell^*})) \leq \epsilon r^*_k(S)$.
Let us now bound the size of $T_{\ell^*}$.  By construction
$|T_{\ell(k)}| \leq k$, and we observe that $T_{\ell^*}$ can be
partitioned into $|T_{\ell(k)}|$ subsets $T_{\ell^*}^u$, for every $u
\in T_{\ell(k)}$, where $T_{\ell^*}^u$ is the set of descendants of
$u$ in $T_{\ell^*}$. The definition of cover tree implies that for
each $u \in T_{\ell(k)}$ and $v \in T_{\ell^*}^u$, $\dist(u,v) \leq
2^{\ell(k)+1}$. Moreover, since the pairwise distance between points
of $T_{\ell^*}^u$ is greater than $2^{\ell^*}$, by applying
Fact~\ref{fact:doubling} with $Y=T_{\ell^*}^u$, $R=2^{\ell(k)+1}$ and
$r=2^{\ell^*}$, we obtain that
\[
|T_{\ell^*}^u| \leq 2^{(\lceil \log_2(8/\epsilon) \rceil +2) \cdot D} 
\leq (64/\epsilon)^D,
\]
and the bound on $|T_{\ell^*}|$ follows. $T_{\ell^*}$ can 
be constructed 
on the explicit tree through a simple level-by-level visit
up to level $\ell^*$, which can be easily determined from $\ell(k)$
and the fact that $\minlv$ is the largest level $\ell$ for which all nodes in
$T_{\ell}$ only have the self-child. The construction time is linear in
\[
\sum_{\ell=\ell^*}^{\maxlv} |T_{\ell}|
= 
\sum_{\ell=\ell^*}^{\ell(k)-1} |T_{\ell}|+\sum_{\ell=\ell(k)}^{\maxlv} |T_{\ell}|
\]
The second summation is clearly upper bounded by
$k \log \Delta$, while, using again Fact~\ref{fact:doubling}
it is easy to argue that $|T_{\ell}| \leq 2^{(\ell(k)+2-\ell)\cdot D}$,
for every $\ell^* \leq \ell \leq \ell(k)-1$, whence the first sum
is $\BO{(64/\epsilon)^D}$. The lemma follows. 
\end{proof}

\noindent
{\bf Remark.} Consider an arbitrary node $u \in T_{\ell^*}$ and recall that,
in the augmented version of the cover tree, the fields
$u.\mbox{weight}$ and $u.\mbox{mis}$ contain, respectively, the size 
and a maximal independent set of $S_u$, where $S_u$ 
is the subset of points of $S$
associated with the descendants of $u$ in $T$. The proof of the above
lemma shows that for any $p \in S_u$ (thus, for any $p$ 
accounted for by $u.\mbox{weight}$ and any $p$ of the maximal independent set)
$\dist(p,u.\mbox{point}) \leq \epsilon r^*_k(S)$.

\subsection{Solving $k$-center} \label{sec:kcenter}
Suppose that
an (augmented) cover tree $T$ for $S$ is available.
We can compute an $\BO{2+\BO{\epsilon}}$-approximate solution $C$ 
to $k$-center on $S$ as follows. First, we extract the coreset $Q = \pts(T_{\ell^*})$,
where $\ell^*=\ell^*(\epsilon,k)$ is the
index defined in Equation~\ref{eq:ellstar}, and then
run  a sequential algorithm 
for $k$-center on $Q$. 
To do so, we could use Gonzalez's 2-approximation algorithm.
However, this would contribute an $\BO{k |Q|}$ term to the running
time, which, based on the size bound stated in
Lemma~\ref{lem:coreset}, would yield a quadratic dependency on $k$.
The asymptotic dependency on $k$ can be lowered by computing the
solution through an adaptation of the techniques presented in
\cite{GoranciHLSS21}, as explained below. Let us define a
generalization of the cover tree data structure, dubbed
\emph{$(\alpha,\beta)$-cover tree}, where the three properties that
each level $\ell$ must satisfy are rephrased as follows:
\begin{enumerate}
\item 
$\pts(T_{\ell}) \subseteq \pts(T_{\ell-1})$;
\item 
for each $u \in T_{\ell}$, $\dist(u, u.{\rm parent}) \leq \beta \cdot \alpha^{\ell+1}$;
\item 
for all $u, v \in T_{\ell}$, $\dist(u, v) > \beta  \cdot \alpha^{\ell}$.
\end{enumerate}
By adapting the insertion procedure of Section~\ref{sec:insertion} and
its analysis, it is easily seen that the insertion of a new point in
the data structure can be supported in $\BO{12^D \cdot \log_{\alpha}
  \Delta}$ time. For a given integer parameter $m$, we construct $m$
generalized cover trees for $Q$, namely an
$(\alpha,\alpha^{p/m})$-cover tree $T^{(p)}$ for every $1 \leq p \leq
m$. Each cover tree is constructed by inserting one point of $Q$ at a
time.  Let $\ell_p$ be the smallest index such that level
$T^{(p)}_{\ell_p}$ in $T^{(p)}$ has at most $k$ nodes. The returned
solution $C$ is the set $\pts(T^{(p)}_{\ell_p})$ such that
$T^{(p)}_{\ell_p}$ minimizes $\alpha^{\ell_p+p/m}$. By selecting
$\alpha=2/\epsilon$ and $m = \BO{\epsilon^{-1} \ln \epsilon^{-1}}$,
and by using the argument of \cite{GoranciHLSS21}, it can be shown
that $C$ is a $(2+\BO{\epsilon})$-approximation for $k$-center on $Q$.

We have:
\begin{theorem}
Given an augmented cover tree $T$ for $S$, the above procedure returns a
 $(2+\BO{\epsilon})$-approximation
$C$ to the $k$-center problem for $S$, 
and can be implemented in time $\BO{(k/\epsilon) (768/\epsilon)^D \log\Delta}$. 
\end{theorem}
\begin{proof}
\sloppy
By Lemma~\ref{lem:coreset}, $Q$ is an
$(\epsilon,k)$-coreset for $S$. Let $C^*=\{c_1, c_2, \ldots, c_k\}$ be
an optimal solution for $k$-center on $S$. The coreset property of $Q$
ensures that for each $c_i$ there is a point $c'_i \in Q$ such that 
$d(c_i,c'_i) \leq \epsilon r^*_k(S)$. This implies that the set
$C'=\{c'_1, c'_2, \ldots, c'_k\}$ is a solution to $k$-center on $Q$
with $r_{C'}(Q) \leq (1+\epsilon) r^*_k(S)$, hence
$r^*_k(Q) \leq (1+\epsilon) r^*_k(S)$.
Suppose that the
$(2+\BO{\epsilon})$-approximation algorithm outlined above is used in
Phase~2 to compute the solution $C$ on $Q$. Then,
$r_C(Q) \leq (2+\BO{\epsilon})r^*_k(Q) \leq 
(2+\BO{\epsilon})(1+\epsilon)r^*_k(S) = (2+\BO{\epsilon}) r^*_k(S)$.
By the coreset property and the triangle inequality, 
it follows that $r_C(S) \leq (2+\BO{\epsilon}) r^*_k$.
For what concerns the
running time, we have that the construction of the
$(\epsilon,k)$-coreset $Q$ requires $\BO{k((64/\epsilon)^D+\log
  \Delta)}$ time (see Lemma~\ref{lem:coreset}), while the running time
of Phase 2 is dominated by the construction of the
$m=\BO{\epsilon^{-1} \ln \epsilon^{-1}}$ $(\alpha,
\alpha^{p/m})$-cover trees $T^{(p)}$, for $1 \leq p \leq m$, by
successive insertions of the elements of $Q$. As observed above, an
insertion takes $\BO{12^D\log_\alpha\Delta}$ time, hence the cost for
constructing each $T^{(p)}$ is $\BO{|Q| 12^D\log_\alpha\Delta}$.
Since $\alpha=2/\epsilon$ and $|Q| \leq k(64/\epsilon)^D$, the total
cost is thus $\BO{(k/\epsilon) (768/\epsilon)^D \log\Delta}$, which
dominates over the cost of Phase 1.
\end{proof}

The constants involved in the analysis are likely to
provide very conservative upper bounds on the actual behavior of the
data structure in practical scenarios. Moreover, for moderate values of
$k$, the use of Gonzalez's algorithm in Phase 2 might prove a much
more practical choice. We wish to remark that
our algorithm improves upon the one
proposed in \cite{GoranciHLSS21} in several ways. First, the accuracy
parameter $\epsilon$ can be chosen freely at query time, and does not
influence the construction of the data structure. Second, our data
structure requires linear space and can handle insertions and
deletions at a lower asymptotic cost.

\subsection{Solving $k$-center with $z$ outliers} \label{sec:outliers}
For the $k$-center problem with $z$ outliers, an approach similar to
the one adopted for $k$-center can be employed.  Let $T$ be an
augmented cover tree for $S$.  We can compute a
$(3+\BO{\epsilon})$-approximation to $k$-center with $z$ outliers on
$S$, by proceeding as follows. First, we extract the coreset $Q =
\pts(T_{\ell^*})$, where $\ell^* = \ell^*(k+z)$ is the index defined
in Equation~\ref{eq:ellstar}. Each point $q \in Q$ is associated to
the weight $w_q = u.\mbox{weight}$, where $u \in T_{\ell^*}$ is such
that $u.\mbox{point}=q$. Then, we extract the solution $C$ from this
weighted coreset $Q$ using the techniques from \cite{CeccarelloPP19},
which are reviewed below.

By Lemma~\ref{lem:coreset}, $Q$ is an $(\epsilon,k+z)$-coreset for $S$
and, based on the remark made after Lemma~\ref{lem:coreset}, all
points of $S$ can be associated to the points of $Q$, such that, for
every $q \in Q$, $w_q$ points of $S$ are associated to $q$ ($q$ is
referred to as the \emph{proxy} for these points) and they are all at
distance at most $\epsilon r^*_{k+z}(S)$ from $q$.  Also recall, from
Equation~\ref{eq:radius-relation}, that $r^*_{k+z}(S) \leq
r^*_{k,z}(S)$. Suppose that algorithm {\sc OutliersCluster} described
in \cite{CeccarelloPP19} is run on the weighted coreset $Q$ with
parameters $k$, $r$, and $\epsilon$, where $r$ is a guess of the
optimal radius. The analysis in \cite{CeccarelloPP19} shows that the
algorithm returns two subsets $X, Q' \subseteq Q$ such that
\begin{itemize}
\item
$|X| \leq k$
\item
For every $p \in S$ whose proxy is in $Q-Q'$,
$\dist(p,X) \leq \epsilon r^*_{k,z}(S)+(3+4\epsilon)r$;
\item
if $r \geq r^*_{k,z}(S)$, then $\sum_{q \in Q'} w_q \leq z$.
\end{itemize}
Then, we can repeatedly run  {\sc OutliersCluster}
for $r = 2^{\maxlv}/(1+\epsilon)^i$, for $i=0,1, \ldots$, stopping at the
smallest guess $r$ which returns a pair $(X,Q')$ where
$Q'$ has aggregate weight at most $z$ and returning $C=X$ as the final solution.We have:
\begin{theorem}
Given an augmented cover tree $T$ for $S$, the above procedure returns a
 $(3+\BO{\epsilon})$-approximation
$C$ to the $k$-center problem with $z$ outliers for $S$, 
and can be implemented in time 
$\BO{(k+z)^2(64/\epsilon)^{2D}(1/\epsilon)\log \Delta}$.
\end{theorem}
\begin{proof}
The bound on the approximation ratio immediately follows from the 
properties of {\sc OutliersCluster} reviewed above. 
The running time is
dominated by the repeated executions of {\sc OutliersCluster}.
Each execution of {\sc OutliersCluster} can be performed in 
$\BO{|Q|^2+k|Q|}=\BO{(k+z)^2(64/\epsilon)^{2D}}$ time. 
The bound on the running time follows by 
observing that $2^{\maxlv}/r^*_{k,z}(S) = \BO{\Delta}$, 
whence $\BO{\log_{1+\epsilon}\Delta}=\BO{(1/\epsilon)\log \Delta}$
executions suffice.
\end{proof}

\subsection{Solving  matroid center}
Consider a matroid $M=(S,I)$ defined on a set $S$, and suppose that
an augmented cover tree $T$ for $S$ w.r.t. $M$ is available.
We can compute an $\BO{3+\BO{\epsilon}}$-approximate solution $C$ to the matroid
center problem on $M$ as follows. First we determine a coreset $Q$
as the union of the independent sets associated with the nodes
of level $T_{\ell^*}$ where $\ell^*=\ell^*(\epsilon,\rank(M))$ is the index defined
in Equation~\ref{eq:ellstar}. Namely,
\[
Q = \bigcup_{u \in T_{\ell^*}} u.\mbox{mis}.
\]
Note that $\rank(M)$ is easily obtained as the size of $r.\mbox{mis}$, where
$r$ is the root of $T$.
Then, solution $C$ is computed by running the 3-approximation algorithm by \cite{ChenLLW16} on $Q$. We have:
\begin{theorem}
\sloppy
Given an augmented cover tree $T$ for $S$ w.r.t. $M=(S,I)$, the above 
procedure returns a
 $(3+\BO{\epsilon})$-approximation
$C$ to the matroid center problem on $M$, 
and can be implemented in time 
$\BO{\mbox{\rm poly}(\rank(M),(64/\epsilon)^{D}) + \rank(M) \log \Delta}$.
\end{theorem}
\begin{proof}
As remarked before, the nodes of $T_{\ell^*}$ induce a partition of
$S$ into subsets $\{S_u : u \in T_{\ell^*}\}$, where $S_u$ is the subset
of points of $S$ associated with the descendants of $u$ in $T$, 
and for each $q \in S_u$ we have $\dist(p,u.\mbox{point}) \leq
\epsilon r^{*}_{{\rm rank}(M)}(S) \leq r^*(M)$ . Consider an optimal solution
$C^*=\{c_1, c_2, \ldots, c_{{\rm rank}(M)}\}$ to the matroid center
problem on $M$, and let $c_i \in S_{u_i}$, for some $u_i \in
T_{\ell^*}$.  We now show that we can substitute each $c_i$ with a
distinct element of $u_i.\mbox{mis} \subseteq Q$, so that the
resulting set of substitutes is also a maximal independent.
Inductively, suppose that we have substituted $c_j$ with a point $c'_j
\in u_j.\mbox{mis}$, for every $1 \leq j < i-1$, and that the set
$C'(i-1)=\{c'_1, \ldots, c'_{i-1}, c_i, \ldots, c_{{\rm rank}(M)}\}$ is
an independent set. By applying Fact~\ref{fact:extAug} with $A=C'(i-1)-\{c_i\}$,
$y=c_i$, $S'=S_{u_i}$, and $B=u_i.\mbox{mis}$, we have that there
exists a point $c'_i \in u_i.\mbox{mis} \backslash C'(i-1)$ such that
$C'(i)=(C'(i-1) \backslash \{c_i\})\cup \{c'_i\}$ is an independent set. 
Let $C'=C'(\rank(M))$.
Since $c_i$ and $c'_i$ belong to the same subset
$S_{u_i}$, we have $d(c_i,c'_i) \leq 2\epsilon r^*(M)$, which immediately
implies that $r_{C'}(Q) \leq (1+2\epsilon)r^*(M)$.
Therefore, the solution $C$ computed using 
the $3$-approximation algorithm by \cite{ChenLLW16}
is such that  $r_{C}(Q) \leq (3+\BO{\epsilon})r^*(M)$.
By the coreset property and the triangle inequality, 
it follows that $r_C(S) \leq (3+\BO{\epsilon}) r^*(M)$.

For what concerns the
running time, we have that the determination of the
level $T_{\ell^*}$ requires 
$\BO{\rank(M)((64/\epsilon)^D+\log \Delta)}$ time (see Lemma~\ref{lem:coreset}),
and the size of the coreset $Q$ is $\BO{(\rank(M))^2(64/\epsilon)^D}$.
The claimed bound follows since the algorithm by \cite{ChenLLW16}
runs in time polynomial in the input size.
\end{proof}

\subsection{Solving diversity maximization}
In \cite{CeccarelloPPU17}, the authors present a coreset-based
approach to approximating the optimal solution of the diversity
maximization problem on a pointset $S$, under all the diversity
measures $\divs(\cdot)$ listed in
Table~\ref{tab:diversity-notions}. Specifically, starting from an
$(\epsilon,k)$-coreset $Q$ for $k$-center on $S$, the authors obtain
the coreset $Q'=Q$ for the remote edge and the remote cycle
variants of diversity maximization, while, for all the other
variants, the coreset $Q'$ is constructed by selecting, for each
$p\in Q$, $\min\{k, |S_p|\}$ points from the subset $S_p$ of a
partition $\{S_p : p\in Q\}$ of $S$ into disjoint subsets, where each
$S_p$ contains points $q\in S$ with $\dist(p,q)\leq \epsilon
r^*_k$. It is shown in \cite{CeccarelloPPU17} that running an $\alpha$
approximation algorithm on $Q'$ yields an $(\alpha
+\BO{\epsilon})$-approximate solution for $S$.  Observe that in all
cases the coreset $Q'$ can be easily constructed from an (augmented)
cover tree $T$ for $S$.  Indeed, in the former, simple case, $Q'$ is
obtained as the set of points associated with the nodes of level
$T_{\ell^*}$, where $\ell^* = \ell^*(\epsilon, k)$ is the index
defined in Equation~\ref{eq:ellstar}. In the latter case, $Q'$ can be
obtained as follows. We need $T$ to be an augmented
cover tree w.r.t. to \emph{k-bounded cardinality matroid} for $S$,
denoted as
$M_{k,S}$, whose independent sets are all subsets of $S$ of at most $k$
points. Then, we simply set $Q' = \cup_{u\in
  T_{\ell^*}}u.\mbox{mis}$. 

For each diversity variant in  Table~\ref{tab:diversity-notions}, let $A_{\rm div}$ be the polynomial-time approximation algorithm yielding the $\alpha_{\rm div}$ approximation mentioned in the table, and let $t_{A_{\rm div}}(\cdot)$
denote its running time. We have:
\begin{theorem}
\sloppy
Consider an cover tree $T$ for $S$ $($augmented w.r.t. the $k$-bounded cardinality matroid $M_{k,S}$, when necessary$)$.  For each diversity variant in  Table~\ref{tab:diversity-notions}, running  $A_{\rm div}$ on the coreset $Q'$ extracted from $T$ returns an $(\alpha_{\rm div}+\BO{\epsilon})$-approximate solution to the diversity maximization problem in time 
$\BO{t_{A_{\rm div}}(k(64/\epsilon)^D)+k\log \Delta}$ for the remote edge and cycle variants, and time $\BO{t_{A_{\rm div}}(k^2(64/\epsilon)^D)+k\log \Delta}$ for the other variants. 
\end{theorem}
\begin{proof}
The stated bounds are an immediate consequence of the above discussion
and the observation that the construction of coreset $Q'$ can be accomplished
in $\BO{k((64/\epsilon)^D+\log \Delta)}$ time, 
for the remote edge and cycle variants, and in
$\BO{k^2(64/\epsilon)^D+k\log \Delta)}$ time, for the other variants. 
\end{proof}

\section{Conclusions} \label{sec:conclusion}
It is important to remark that for all problems treated in this
paper, when $S$ is large and both the spread $\Delta$
and the doubling dimension $D$  of the metric are small, the dynamic maintenance 
of the augmented cover tree data structure  
and the extraction of solutions can be accomplished in
time dramatically smaller than the time that would be required to
compute solutions on the entire pointset from scratch. 

Finally, the coreset-based approaches developed in \cite{PietracaprinaPS20}
for robust matroid center, and in \cite{CeccarelloPP20} for diversity maximization under matroid constraints, can be integrated with the approach presented
in this paper, to yield dynamic maintenance for these more general versions of the problems, with similar accuracy-performance tradeoffs. 

%
%
\bibliographystyle{splncs04}
\bibliography{references}

\begin{thebibliography}{10}
\providecommand{\url}[1]{\texttt{#1}}
\providecommand{\urlprefix}{URL }
\providecommand{\doi}[1]{https://doi.org/#1}

\bibitem{AbbassiMT13}
Abbassi, Z., Mirrokni, V.S., Thakur, M.: Diversity maximization under matroid
  constraints. In: Proc. {ACM KDD}. pp. 32--40 (2013)

\bibitem{Bateni2023Optimal-Fully-D}
Bateni, M., Esfandiari, H., Fichtenberger, H., Henzinger, M., Jayaram, R.,
  Mirrokni, V., Wiese, A.: Optimal fully dynamic k-center clustering for
  adaptive and oblivious adversaries. In: Proc. {ACM-SIAM SODA}. pp. 2677--2727
  (2023)

\bibitem{BeygelzimerKL06}
Beygelzimer, A., Kakade, S., Langford, J.: Cover trees for nearest neighbor.
  In: Proc. {ICML}. pp. 97--104 (2006)

\bibitem{Borassi2019Better-sliding-}
Borassi, M., Epasto, A., Lattanzi, A., Vassilvitskii, S., Zadimoghaddam, M.:
  Better sliding window algorithms to maximize subadditive and diversity
  objectives. In: Proc. {ACM PODS}. pp. 254--268 (2019)

\bibitem{CeccarelloPP19}
Ceccarello, M., Pietracaprina, A., Pucci, G.: Solving k-center clustering (with
  outliers) in mapreduce and stre aming, almost as accurately as sequentially.
  {PVLDB}  \textbf{12}(7),  766--778 (2019)

\bibitem{CeccarelloPP20}
Ceccarello, M., Pietracaprina, A., Pucci, G.: A general coreset-based approach
  to diversity maximization under matroid constraints. {ACM} Trans. Knowl.
  Discov. Data  \textbf{14}(5),  60:1--60:27 (2020)

\bibitem{CeccarelloPPU16}
Ceccarello, M., Pietracaprina, A., Pucci, G., Upfal, E.: A practical parallel
  algorithm for diameter approximation of massive weighted graphs. In: Proc.
  {IEEE} IPDPS. pp. 12--21 (2016)

\bibitem{CeccarelloPPU17}
Ceccarello, M., Pietracaprina, A., Pucci, G., Upfal, E.: Mapreduce and
  streaming algorithms for diversity maximization in metric spaces of bounded
  doubling dimension. Proc. {VLDB} Endow.  \textbf{10}(5),  469--480 (2017)

\bibitem{Chakrabarty2020The-Non-Uniform}
Chakrabarty, D., Goyal, P., Krishnaswamy, R.: The non-uniform {$k$}-center
  problem. {ACM Trans. on Algorithms}  \textbf{16}(4),  46:1--46:19 (2020)

\bibitem{Chan2018Fully-dynamic-k}
Chan, T.H., Guerqin, A., Sozio, M.: Fully dynamic k-center clustering. In:
  Proc. {WWW}. pp. 579--587 (2018)

\bibitem{Chan2023Fully-Dynamic-k}
Chan, T.H., Lattanzi, S., Sozio, M., Wang, B.: Fully dynamic k-center
  clustering with outliers. In: Proc. {COCOON}. pp. 150--161 (2023)

\bibitem{Charikar1997Incremental-clu}
Charikar, M., Chekuri, C., Feder, T., Motwani, R.: Incremental clustering and
  dynamic information retrieval. In: Proc. {ACM STOC}. pp. 626--635 (1997)

\bibitem{Charikar2001}
Charikar, M., Khuller, S., Mount, D., Narasimhan, G.: {Algorithms for Facility
  Location Problems with Outliers}. In: Proc. {ACM-SIAM SODA}. pp. 642--651
  (2001)

\bibitem{ChenLLW16}
Chen, D., J.Li, H.L., Wang, H.: Matroid and knapsack center problems.
  Algorithmica  \textbf{75}(1),  27--52 (2016)

\bibitem{ChiplunkarKR20}
Chiplunkar, A., Kale, S., Ramamoorthy, S.: How to solve fair k-center in
  massive data models. In: Proc. {ICML}. pp. 1877--1886 (2020)

\bibitem{Cohen-Addad2016Diameter-and-k-}
Cohen-Addad, V., Schwiegelshohn, C., Sohler, C.: Diameter and k-center in
  sliding windows. In: Proc. {ICALP} (2016)

\bibitem{Elkin2022Counterexamples}
Elkin, Y., Kurlin, V.: Counterexamples expose gaps in the proof of time
  complexity for cover trees introduced in 2006. In: Proc. TopoInVis. pp.
  9--17. IEEE (2022)

\bibitem{Gonzalez85}
Gonzalez, T.: {Clustering to Minimize the Maximum Intercluster Distance}.
  {Theoretical Computer Science}  \textbf{38},  293--306 (1985)

\bibitem{GoranciHLSS21}
Goranci, G., Henzinger, M., Leniowski, D., Schulz, C., Svozil, A.: Fully
  dynamic \emph{k}-center clustering in low dimensional metrics. In: Proc.
  {ALENEX} 2021. pp. 143--153 (2021)

\bibitem{GottliebKK14}
Gottlieb, L.A., Kontorovich, A., Krauthgamer, R.: Efficient classification for
  metric data. {IEEE} Trans. Information Theory  \textbf{60}(9),  5750--5759
  (2014)

\bibitem{Gupta2003Bounded-geometr}
{Gupta}, A., {Krauthgamer}, R., {Lee}, J.R.: Bounded geometries, fractals, and
  low-distortion embeddings. In: Proc. {IEEE FOCS}. pp. 534--543 (2003)

\bibitem{Harris2019A-Lottery-Model}
Harris, D., Pensyl, T., Srinivasan, A., Trinh, K.: A lottery model for
  center-type problems with outliers. {ACM} Trans. on Algorithms
  \textbf{15}(3),  36:1--36:25 (2019)

\bibitem{Kale19}
Kale, S.: Small space stream summary for matroid center. In: Proc.
  {APPROX/RANDOM}. pp. 20:1--20:22 (2019)

\bibitem{KleindessnerAM19}
Kleindessner, M., Awasthi, P., Morgenstern, J.: Fair k-center clustering for
  data summarization. In: Proc. {ICML}. pp. 3448--3457 (2019)

\bibitem{KrauthgamerL04}
Krauthgamer, R., Lee, J.: Navigating nets: simple algorithms for proximity
  search. In: Proc. {ACM-SIAM SODA}. pp. 798--807 (2004)

\bibitem{Leskovec2020Mining-of-massi}
Leskovec, J., Rajaraman, A., Ullman, J.: Mining of Sassive Data Sets. Cambridge
  University Press (2014)

\bibitem{Oxley06}
Oxley, J.: {Matroid Theory}. Oxford University Press (2006)

\bibitem{PellizzoniPP20}
Pellizzoni, P., Pietracaprina, A., Pucci, G.: Dimensionality-adaptive k-center
  in sliding windows. In: Proc. {DSAA}. pp. 197--206 (2020)

\bibitem{Pellizzoni2022k-Center-Cluste}
Pellizzoni, P., Pietracaprina, A., Pucci, G.: k-center clustering with outliers
  in sliding windows. Algorithms  \textbf{15}(2), ~52 (2022)

\bibitem{PietracaprinaPS20}
Pietracaprina, A., Pucci, G., Sold{\`{a}}, F.: Coreset-based strategies for
  robust center-type problems. arXiv  \textbf{2002.07463} (2020)

\end{thebibliography}

\appendix
\section*{Appendix}

\subsection*{Pseudocode for the insert procedure}
Algorithm~\ref{alg:insert} details the pseudocode of the insertion procedure.  It takes in input the point $p$ to be inserted and the cover tree, represented by its root $r$.

\begin{algorithm}[h]
\scriptsize
\SetAlgoLined
\If{$\dist(p, r) > 2^{\ell_{\rm max}}$}{
	$\ell_{\rm max} = \lfloor \log_2 \dist(p, r) \rfloor$ \\
	$r.\mbox{level} = \ell_{\rm max}$ \\
}
let $\ell = \ell_{\rm max}$ \\
$Q_{\ell} = \{ r \}$ \\
\While{$Q_{\ell} \neq \emptyset$}{
	$Q_{\ell-1} = \emptyset$ \\
	\For{$q\in Q_{\ell}$}{
		\If{\rm $q.\mbox{children} == \emptyset$ OR $q.\mbox{children}[0].\mbox{level != } \ell-1$}{
                \If{$\dist(q, p) \leq 2^\ell$}{
                    $Q_{\ell-1} = Q_{\ell-1} \cup \{ q \}$ \\
                }
		}
		\Else{
			$Q_{\ell-1} = Q_{\ell-1} \cup \{ q' \in q.\mbox{children}$ s.t $\dist(p, q') \leq 2^{\ell} \}$ \\
		}
    	}
    	$\ell = \ell -1$ \\
}
\While{$\dist(p, Q_{\ell+1}) > 2^{\ell+1}$}{
	$\ell = \ell+1$ \\
}
let $v \in Q_{\ell+1}$ be s.t. $\dist(p, v) \leq 2^{\ell+1}$ \\
$u = $ new node with $u.\mbox{point}  = p$, $u.\mbox{level}  = \ell$ \\
\If{\rm $v.\mbox{children} ==\emptyset$ OR 
$v.\mbox{children}[0].\mbox{level != } \ell$}{
$w = $ new node with $w.\mbox{point} = v.\mbox{point}$,  $w.\mbox{level}  = \ell$, $w.\mbox{children}  = v.\mbox{children}$ \\
	$v.\mbox{children} = \{ u, w \}$ \\
}
\Else{
	$v.\mbox{children}  = v.\mbox{children}  \cup \{ u \}$ \\
}
$t = u$ \\
\While{\rm $t \neq$ null}{
	$t.\mbox{weight} = t.\mbox{weight}+1$ \\
	add $p$ to $t.\mbox{mis}$ if it remains an independent set \\
	$t = r.\mbox{parent}$ \\
}

\caption{\proc{Insert}(Point $p$, Root $r$)} \label{alg:insert}
\end{algorithm}

\subsection*{Pseudocode for the delete procedure}
Algorithm~\ref{alg:delete} details the pseudocode of the deletion procedure.  It takes in input the point $p$ to be deleted and the cover tree, represented by its root $r$.

\begin{algorithm}[h]
\scriptsize
\SetAlgoLined
let $\ell = \ell_{\rm max}$ \\
$Q_{\ell} = \{ r \}$ \\
\While{\rm true}{
	$Q_{\ell-1} = \emptyset$ \\
	\For{$q\in Q_{\ell}$}{
            \If{\rm $q.\mbox{point}  == p$ AND $q.\mbox{children}  == \emptyset$}{
                \textbf{break while} \\
            }
		\If{\rm $q.\mbox{children}  == \emptyset$ OR $q.\mbox{children}[0].\mbox{level != } \ell-1$}{
                \If{$\dist(q, p) \leq 2^\ell$}{
                    $Q_{\ell-1} = Q_{\ell-1} \cup \{ q \}$ \\
                }
		}
		\Else{
			$Q_{\ell-1} = Q_{\ell-1} \cup \{ q' \in q.\mbox{children}$ s.t $\dist(p, q') \leq 2^{\ell} \}$ \\
		}
        }
        $\ell = \ell -1$ \\
}
$R_\ell = \emptyset$ \\
\While{$\ell \leq \ell_{\rm max}-1$}{
    \If{\rm $\exists u \in Q_\ell$ s.t. $u.\mbox{point}  == p$ and $u.\mbox{level} == \ell$}{
        $v = u.\mbox{parent}$ \\
        delete $u$ from $v.\mbox{children}$ \\
        \If{\rm $v.\mbox{point}  == p$}{
            $R_\ell = R_\ell \cup v.\mbox{children}$ \\
        }
        \uElseIf{\rm $|v.\mbox{children}| == 1$}{
            $v.\mbox{children}  = v.\mbox{children}[0].\mbox{children}$ \\
            delete $v.\mbox{children}[0]$ \\
        }
    }
    let $R_{\ell+1} = \emptyset$ \\
    \For{$w \in R_\ell$}{
        \If{\rm $\exists w' \in Q_{\ell+1} \cup R_{\ell+1}$ s.t. $\dist(w, w') \leq 2^{\ell+1}$}{
            \If{\rm $w'.\mbox{children} \neq \emptyset$ AND $w'.\mbox{children}[0].\mbox{level} == \ell$}{
                $w'.\mbox{children}  = w'.\mbox{children} \cup \{ w \}$ \\
            }
            \Else{
                $z = $ new node with $z.\mbox{point}  = w'.\mbox{point}$, $z.\mbox{level}  = \ell$, \\
\hspace*{0.45cm} $z.\mbox{children}  = w'.\mbox{children}$ \\
                $w'.\mbox{children}  = \{ z, w \}$ \\
            }
        }
        \Else{
            $w.\mbox{level}  = \ell +1$\\
            $R_{\ell+1} = R_{\ell+1} \cup \{ w \}$ \\
        }
    }
    \For{$w \in Q_{\ell+1} \cup R_{\ell+1}$}{
        update $w.\mbox{weight}$ and $w.\mbox{mis}$ \\
    }
    $\ell = \ell + 1$
}
\If{\rm $p \neq r.\mbox{point}$}{
    \If{$R_\ell == \emptyset$}{
        $\ell_{\rm max} = r.\mbox{children}[0].\mbox{level} + 1$ \\
        $r.\mbox{level}  = \ell_{\rm max}$ \\
    }
    \Else{
        $r_{\rm new} = $ new root with $r_{\rm new}.\mbox{point} = r.\mbox{point}$, $r_{\rm new}.\mbox{level} = \ell + 1$, \\
\hspace*{0.85cm} $r_{\rm new}.\mbox{children} = \{ r \} \cup R_{\ell}$ \\ 
    }
}
\Else{
    \If{$|R_\ell| == 1$}{
        let $v = R_\ell[0]$ be the new root \\
        $\ell_{\rm max} = r$.children$[0]$.level $+ 1$ \\
        $v.\mbox{level} = \ell_{\rm max}$ \\
        update $v.\mbox{weight}$ and $v.\mbox{mis}$ \\
    }
    \Else{
        let $v \in R_\ell$ \\
        $r_{\rm new} = $ new root with $r_{\rm new}.\mbox{point} = v.\mbox{point}$, $r_{\rm new}.\mbox{level} = \ell + 1$, $r_{\rm new}.\mbox{children} = R_{\ell}$ \\ 
        $\ell_{\rm max} = \ell +1$ \\
        update $r_{\rm new}.\mbox{weight}$ and $r_{\rm new}.\mbox{mis}$ 
    }
}
\caption{\proc{Delete}(Point $p$, Root $r$)} \label{alg:delete}
\end{algorithm}
\end{document}